\documentclass[letterpaper, 10 pt, conference]{ieeeconf}  
\IEEEoverridecommandlockouts                              % This command is only needed if 
\usepackage{amssymb}

\usepackage{amsthm}
\let\labelindent\relax
\usepackage{enumitem}
\usepackage{graphicx}               % Necessary to use \scalebox

\theoremstyle{plain}
\newtheorem{theorem}{Theorem}
\newtheorem{lemma}{Lemma}

\newtheorem{corollary}{Corollary}
\theoremstyle{definition}
\newtheorem{definition}{Definition}
\newtheorem{assumption}{Assumption}
\newtheorem{problem}{Problem}

\theoremstyle{remark}
\newtheorem{remark}{Remark}

\usepackage{mathrsfs}          % \mathscr
\usepackage{mathtools}         % extends amsmath
\usepackage{bm}                % \bm for bold-math
\usepackage{dsfont}            % \mathds{1} indicator
\usepackage{centernot}         % \centernot
\usepackage[normalem]{ulem}    % \uline etc. (normalem keeps \emph)
\usepackage[export]{adjustbox}
\usepackage[ruled,vlined]{algorithm2e}
\usepackage{subcaption}
\usepackage[space]{grffile}    % filenames with spaces
\usepackage{svg}
\usepackage{pdfsync}

\definecolor{darkblue}{RGB}{0,0,102}
\definecolor{lightblue}{RGB}{77,77,148}
\definecolor{gold}{RGB}{234,170,0}
\definecolor{metallic_gold}{RGB}{139,111,78}

\renewcommand{\rm}[1]{\textrm{#1}}   % override LaTeX2e default

\newcommand{\lmat}{\begin{bmatrix}}
\newcommand{\rmat}{\end{bmatrix}}

\DeclareMathOperator*{\argmin}{argmin}

\newcommand{\R}{\mathbb{R}}
\newcommand{\E}{\mathbb{E}}            % expectation
\newcommand{\Prob}{\mathbb{P}}
\newcommand{\Filt}{\mathcal{F}}
\newcommand{\gen}{\mathcal{L}}          % infinitesimal generator
\newcommand{\intv}{\mathcal{M}}         % intervention operator
\newcommand{\Guard}{\Gamma}
\newcommand{\Xspace}{\mathcal{X}}
\newcommand{\Xtgt}{\mathcal{X}_{t}}
\newcommand{\Xuns}{\mathcal{X}_{u}}
\newcommand{\Xini}{\mathcal{X}_{0}}

\newcommand{\st}{\psi}                  % the stopping time
\newcommand{\sep}{\chi}

\newcommand{\Yproc}{Y} 
\newcommand{\tmin}{\tau_{\min}}
\newcommand{\amax}{A_{\max}}

\newcommand{\jumps}{j_{\st}}% the frozen process
\newcommand{\sbar}{\bar{\sigma}}
\newcommand{\ind}[1]{\mathbf{1}\!\left[#1\right]}

\newcommand{\Flow}{\mathcal{C}}

\pdfoutput=1

\title{
\LARGE \bf
Reach-Avoid Certificates for Contact-based Stochastic Hybrid Systems \vspace{-0.5em}}

\author{Tianyi Xia$^1$ and Ryan K. Cosner$^2$
\thanks{$^{1}$ Xia is with GRASP Lab, University of Pennsylvania, PA 19104, USA. Email: \texttt{tianyxia@engineering.upenn.edu}}
\thanks{$^{2}$ Cosner is with the Department of Mechanical Engineering, Tufts University, 419 Boston Avenue, Medford, MA 02155, USA. Email: \texttt{ryan.cosner@tufts.edu}}
}

\begin{document}

\maketitle

\thispagestyle{empty}
\pagestyle{empty}

%%%%%%%%%%%%%%%%%%%%%%%%%%%%%%%%%%%%%%%%%%%%%%%%%%%%%%%%%%%%%%%%%%%%%%%%%%%%%%%%
\begin{abstract}
% \todo{
% \begin{itemize}
%     \item 1 quick motivation sentence
%     \item $\sim $ 2 sentences of novel method
%     \item 1 sentence of contribution relative to existing methods. 
%     \item 1 sentence pointing towards applications, code base, or links to videos
% \end{itemize}
% }
This work studies certified reach-avoid control of stochastic hybrid systems in which a stochastic differential equation governs continuous flow and a controlled stochastic kernel governs guard-triggered resets. The input appears only in the reset kernel, where it determines the post-jump distribution. Certificates for such systems have so far assumed spontaneous or controller-chosen jump times, a finite horizon, or discrete-time dynamics; we certify, over the infinite horizon, that the target is reached before the unsafe set with at least a prescribed probability. First, we introduce reach-avoid contact certificates, whose flow and
jump conditions are imposed only on a sublevel corridor, and prove that
any such certificate with a compatible reset policy meets the
specification; a pointwise safety filter extends this guarantee from verification to control by constraining arbitrary
nominal controllers. Next, we prove that executions are non-Zeno and bound the expected settling time and reset count.
We frame certificate search as a linear program over any twice-differentiable dictionary 
followed by independent verification. We illustrate our approach via comparisons where the certified region is a subset covering 93.6\% of the ground truth and where our certificate-based safety filter reaches the goal while reducing collisions from 68 to zero across 200 episodes of a dribbling ball problem. 
\end{abstract}

%%%%%%%%%%%%%%%%%%%%%%%%%%%%%%%%%%%%%%%%%%%%%%%%%%%%%%%%%%%%%%%%%%%%%%%%%%%%%%%%
\section{INTRODUCTION}
Stochastic hybrid systems (SHS), in which a diffusion is interrupted by discrete transitions that redraw the state from a stochastic kernel, model plants with impacts, mode switches, and event-triggered actuation \cite{hu2000towards,bujorianu2006toward,teel2014survey}.
Guarantees bounding the probability of failure for such systems have been obtained from barrier
and supermartingale inequality certificates on the transitions without discretizing the state space or computing reachable sets \cite{kushner1967stochastic,prajna2007framework,lavaei2022shs}. The
transition inequality, however, is available only under assumptions on
when the transitions occur. Existing formulations take the transition times to be spontaneous, driven by a Poisson process or a state-dependent rate \cite{prajna2007framework}, scheduled by a timer \cite{lavaei2022shs}, or selected by the controller, as in impulse control \cite{bensoussan1984impulse}. None yield a transition inequality for state-based forced transitions, the other standard jump mechanism of general stochastic hybrid systems \cite{hu2000towards,bujorianu2006toward}. There the jump is triggered by the state reaching a guard, so its timing is dictated by the continuous trajectory rather than by a clock or a controller. %, while its outcome is still drawn from a stochastic kernel.

The class considered in this paper carries its entire control authority in these transitions. The flow is an autonomous stochastic differential equation; on reaching a guard, the state is redrawn from a kernel shaped by the input, so the controller is idle between transitions, and at a transition, it shapes the distribution of the post-impact state. 
The problem specification is \textit{reach-avoid}, i.e., to reach a target set before entering an unsafe set, and the guarantee sought is a lower bound, valid over the infinite horizon, on the probability that a declared reset policy achieves it.

Three features of this class of SHS place it outside the existing theory. (1)
the transition times are decided by the state, so the inequality
across a transition must hold at every point of the guard that the
flow can reach, (2) the input is confined to the reset map, so the certificate must
also shape the continuous flow between transitions, and (3)  a guarantee over the infinite
horizon presupposes that the transitions do not accumulate infinitely in finite
time (i.e., are non-Zeno); for a stochastic plant this cannot be assumed, since the dwell
times are random, and their infimum may be zero.

Our certificate is a function $V$ that decreases at a fixed rate along the continuous flow and does not increase in expectation across resets, both required only on the sublevel corridor $\{~x~|~V(x)<\lambda~\}$ where % together with 
$V(x)\ge\lambda$ on the unsafe set and $V\le1$ on the initial set for $\lambda > 1$. Freezing $V(x_t)$ when $x_t$ leaves this corridor generates a supermartingale and Ville's inequality gives the desired infinite-horizon probability bound. We prove the latter from a uniform margin by which every reset leaves the guard and from bounds on the drift and diffusion of the guard function. 

The contributions of this work are threefold. 
Firstly, we introduce reach-avoid contact certificates and prove that, with a compatible reset policy, they guarantee the specification over the infinite horizon with probability at least $1-1/\lambda$. The same certificate is then used to shield an arbitrary base policy through a pointwise check at each reset. Secondly, we prove that executions are non-Zeno under two conditions on the plant, and we derive bounds on the expected time and number of contacts until the sublevel corridor is left. Finally, we show that the certificate conditions are linear in $V$, so
  that the candidate certificate search can be framed as a linear program over any twice-differentiable
  dictionary and verified by interval branch and bound, and we provide three example deployments of our certificate.

\subsection{Related Work}

\emph{Stochastic hybrid systems.} General models that combine a
diffusion with forced (guard-triggered) and spontaneous jumps were set
out in \cite{hu2000towards,bujorianu2006toward}; setting
$g\equiv0$ recovers the controlled piecewise-deterministic Markov
processes of \cite{davis1984pdmp}, and the Lyapunov theory for
stability of this family is surveyed in \cite{teel2014survey}. Safety
certificates for stochastic hybrid systems have so far treated jumps
that arrive at a state-dependent rate or through a Poisson process
\cite{prajna2007framework,lavaei2022shs}, or jump times chosen by the
controller as in impulse control \cite{bensoussan1984impulse}. A reset
that is forced by contact with a guard and whose outcome is random fits neither
setting and, to our knowledge, is not covered by any existing reach-avoid certificate.

\emph{Supermartingale certificates.} The use of supermartingales to
bound the probability that a stochastic system leaves a set goes back
to \cite{kushner1967stochastic}. Sum-of-squares programming made such
certificates computable for polynomial systems
\cite{prajna2007framework}, finite-horizon variants relax the
supermartingale condition to a bounded growth rate
\cite{steinhardt2012finite,santoyo2021barrier}, and they have enabled control barrier
functions \cite{ames_control_2017} to be applied to diffusions
\cite{clark2021control}. Neural parametrizations remove the polynomial
restriction in discrete time \cite{mathiesen2023safety} and, for
reach-avoid-stay specifications of diffusions, in continuous time
\cite{neustroev2025neural}.
Reach-avoid supermartingales for discrete-time systems are learned in
\cite{zikelic2023learning}, and discrete-time barrier functions under
random disturbances are analyzed in
\cite{cosner23robust,cosner_bounding_2024}. All of these treat either
a diffusion without jumps or a sampled system with a fixed step; in
our setting the time between contacts is itself random and has to be
lower bounded before any martingale argument is applied.

\emph{Reachability and dynamic programming.} The reach-avoid probability of a diffusion is characterized by a partial differential
equation \cite{esfahani2016stochastic}, and that of a discrete-time
stochastic hybrid system by dynamic programming
\cite{abate2008probabilistic,summers2010verification}; both give exact
values but require a state grid. Hamilton-Jacobi reachability
\cite{bansal2017hj} has recently been extended to hybrid systems with
controlled and forced transitions \cite{borquez2026hj} under bounded
adversarial disturbances over a finite horizon on a state grid.

\emph{Hybrid systems with impacts.} Deterministic systems with impulse
effects are the standard model of legged locomotion
\cite{grizzle2001asymptotically,westervelt2003hybrid}, and
\cite{goebel2012hybrid} develops the general stability theory. Zeno
executions, in which contacts accumulate in finite time, are
characterized in \cite{lygeros2003dynamical} and are usually excluded
by assumption. In a stochastic plant, that assumption is not available,
and we prove the property from the reset margin and the bounds on the
guard function instead.

\section{Problem Formulation}\label{sec:background}

We consider stochastic hybrid systems\footnote{\emph{Notation.}
$\R_{\ge0}$ denotes the nonnegative reals, 
$\mathcal B(\cdot)$ the Borel
$\sigma$-algebra, and $\ind{\cdot}$ the indicator. A \emph{stochastic
kernel} on a Borel space $S$ given $T$ is a map $Q(\cdot\mid\cdot)$ that
is a probability measure on $\mathcal B(S)$ for each fixed element of
$T$, and Borel measurable in that element for each fixed
$B\in\mathcal B(S)$. All processes are defined on a filtered probability
space $(\Omega,\Filt,\Prob,(\Filt_t)_{t\ge0})$ satisfying the usual
conditions. We write $\Prob_{x_0}$ and $\E_{x_0}$ for the law and
expectation of the process started at $x_0$.}.
$\Sigma$ on a state space
$\Xspace\subseteq\R^{l}$ in which a diffusion is interrupted by
guard-triggered \emph{controlled random resets}:
\begin{equation}\label{eq:shs}
  \begin{cases}
    \mathrm{d}x_t = f(x_t)\,\mathrm{d}t + g(x_t)\,\mathrm{d}W_t,
      & x_t\in \Flow,\\[2pt]
    x^{+}\sim Q(\cdot\mid x^{-},u), & x^{-}\in\Guard,
  \end{cases}
\end{equation}
with $W$ a standard Wiener process. The flow set and guard are the level
sets of a scalar function $h:\Xspace\to\R$,
\begin{equation}\label{eq:guard}
  \Flow:=\{x:\ h(x)>0\},\qquad \Guard:=\{x:\ h(x)=0\},
\end{equation}
where $h\in C^{4}$ and $0$ is a regular value of $h$, so $\Guard$ is an
embedded codimension-one submanifold. The coefficients $f,g\in C^{2}$ are
Lipschitz with linear growth, ensuring a pathwise-unique strong solution
between resets \cite{oksendal2003sde}. On hitting $\Guard$ at pre-jump state
$x^{-}$, the state is redrawn from the stochastic kernel $Q$ under a
control input $u\in A(x^{-})$, where $A:\Guard\rightrightarrows U$ is a
Borel measurable action map with nonempty values into a Borel space $U$.
Both channels are stochastic: the flow through $g\,\mathrm{d}W_t$ and the
reset through~$Q$.

\subsection{Executions and operators}\label{sec:bg-exec}

Executions alternate between flows and resets. Define \emph{contact times}
\begin{equation*}\label{eq:contact-times}
  \tau_1:=\inf\{t\ge0: x_t\in\Guard\},\quad
  \tau_{k+1}:=\inf\{t>\tau_k: x_t\in\Guard\},
\end{equation*}
with $\inf\varnothing:=\infty$; define $x^{-}_k:=x_{\tau_k^{-}}\in\Guard$ and 
$x^{+}_k:=x_{\tau_k}$ for the pre- and post-jump states (so $x_t$ is
right-continuous), and $\Delta t_k:=\tau_{k+1}-\tau_k$ for the $k$-th
dwell time. A Borel set $S$ is \emph{entered} at time $t$ if $x_t\in S$
or $x_{t^{-}}\in S$ (with $x_{0^{-}}:=x_0$), so that a contact counts
through its pre-jump state as well as its post-jump state; the
\emph{entering state} at an entrance time is $x_{t^{-}}$ if
$x_{t^{-}}\in S$ and $x_t$ otherwise, and reaching a set always means
entering it in this sense. Each $\tau_k$ is a stopping time by the
debut theorem \cite{dellacherie1978probabilities}, as is every entrance
time of a Borel set appearing below. The number of resets is \textit{a priori} unbounded;
executions are globally defined precisely when $\sum_k\Delta t_k=\infty$
almost surely.

Certificates are time-invariant functions $V:\Xspace\to\R$; for
$V\in C^{2}$ the flow and the resets act through the generator and the
\emph{intervention operator} denoted as:
\begin{align}
  \gen V &:= \nabla V\cdot f
     +\tfrac12\operatorname{tr}\!\big(gg^{\top}\nabla^{2}V\big),
     \label{eq:generator}\\
  \intv V(x^{-}) &:= \inf_{u\in A(x^{-})}\
     \E_{x^{+}\sim Q(\cdot\mid x^{-},u)}\big[V(x^{+})\big],
     \label{eq:intervention}
\end{align}
which are the continuous and discrete halves of a quasi-variational inequality
\cite{bensoussan1984impulse}.

\subsection{Problem statement}\label{sec:bg-problem}

The specification consists of Borel sets $\Xini$ (initial), $\Xtgt$
(target) and $\Xuns$ (unsafe), with $\Xtgt$ closed and $\Xini\cap\Xuns=\Xtgt\cap\Xuns=\varnothing$. For confidence $p\in[0,1)$, set
  $\lambda:=\tfrac{1}{1-p}\in(1,\infty).$

\begin{problem}[Certified reach-avoid]\label{prob:main}
Given the SHS $\Sigma$, specification sets $(\Xini,\Xtgt,\Xuns)$ and a
confidence level $p\in[0,1)$, produce a reset policy $ u = \pi(x^-)$ together
with a certificate guaranteeing event $\mathcal{A}$, %of the fact 
for every $x_0\in\Xini$, the
closed loop reaches $\Xtgt$ before entering $\Xuns$, with probability
at least $p$:
\begin{center}
    $\mathcal{A} = \big\{\exists\,T<\infty:\ \Xtgt\ \text{is entered at }T\
  \text{and}\ x_t\notin\Xuns, \forall t<T\big\}$\\
  $\mathbb{P}_{x_0} [\mathcal{A}] \geq p $
\end{center}
\end{problem}

\subsection{Assumptions}\label{sec:bg-assumptions}

\begin{assumption}[Transversality and non-Zeno data]
\label{asm:wellposed}
Every reset must leave the guard, % behind, 
either in position or in
velocity.
\begin{description}[leftmargin=1.0em, labelindent=0.5em]
  \item[(B1)] \emph{Uniform separation.} There is $\delta>0$ such that
    one of the following holds for all $x^{-}\in\Guard$,
    $u\in A(x^{-})$ and $Q(\cdot\mid x^{-},u)$-a.e.\ $x^{+}$:
    (i)~$h(x^{+})\ge\delta$, the reset lands at a positive distance
    from the guard; or (ii)~$\nabla h^{\top}g\equiv0$ on $\Xspace$ and
    $\gen h(x^{+})\ge\delta$, the reset relaunches from a
    configuration-level guard at a positive rate, the kernel being
    truncated away from tangential relaunch.

  \item[(B2)] \emph{Bounded data.} For finite constants $\amax,\sbar$,
    on $\Xspace$: $|\gen h|\le\amax$ and $\|\nabla h^{\top}g\|\le\sbar$
    in case \textup{(i)}; $|\gen(\gen h)|\le\amax$ and
    $\|\nabla(\gen h)^{\top}g\|\le\sbar$ in case \textup{(ii)}.
\end{description}
\end{assumption}

\begin{assumption}[Certificate and policy regularity]
\label{asm:regularity}
\leavevmode
\begin{description}[leftmargin=1em, labelindent=0.5em]
  \item[(R1)] $V\in C^{2}$ on $\Xspace$.
  \item[(R2)] $\E_{Q}[V(x^{+})| x^{-}_,u]<\infty$ for every
    $x^{-}\in\Guard$, $u\in A(x^{-})$.
    \item[(R3)] A \emph{reset policy} is a Borel measurable map
  $\pi:\Guard\to U$ with $\pi(x^{-})\in A(x^{-})$.
\end{description}
\end{assumption}

\begin{remark}[Role of the assumptions]\label{rem:asm-roles}
The dwell-time analysis of Lem.~\ref{lem:non-zeno} is a first-passage
problem for the scalar separated by (B1), which starts every excursion
above $\delta$ and must reach zero before the excursion can end. In
case (i) that scalar is $s_t:=h(x_t)$. % and the claim is
% immediate. 
In case (ii) it is $\dot s_t$, and
$\nabla h^{\top}g\equiv0$ is what makes
the reduction valid: by It\^o's formula
$\mathrm{d}s_t=\gen h(x_t)\,\mathrm{d}t$ with no martingale part, so
$s_t$ is $C^{1}$ along every excursion with $\dot s_t=\gen h(x_t)$ and
cannot return to zero before $\dot s_t$ does. Together, (B1) and (B2)
supply the constants
$\eta$, $p_0$ and $\tmin$. For the certificate, (R1) makes \eqref{eq:generator} well defined and (R2)
makes \eqref{eq:intervention} so, while (R3) ensures that no measurable-selection argument is required. 
\end{remark}

\section{Reach-Avoid Safety Guarantees for a Class of Stochastic Hybrid
Systems}\label{sec:guarantees}

Our main result is that a certificate in the sense of
Definition~\ref{def:certificate}, together with any compatible reset
policy, yields an infinite-horizon probabilistic reach-avoid guarantee
for the closed loop system. For a candidate $V$ and level $\lambda$, define
\begin{equation}\label{eq:owed}
  \Gamma_\lambda:=\{x\in\Guard\setminus\Xtgt:\ V(x)<\lambda\},
\end{equation}
and let $\st$ be the entrance time of $\Xtgt\cup\{V\ge\lambda\}$,
\begin{center}
\scalebox{0.9}{
$
  \st:=\inf\big\{t\ge0: x_t\in\Xtgt\cup\{V\ge\lambda\}\ \text{or}\
  x_{t^{-}}\in\Xtgt\cup\{V\ge\lambda\}\big\},
$
}    
\end{center}
and, on $\{\st<\infty\}$, let $\xi$ denote its entering state. The
certificate settles the run at $\st$: either $\xi\in\Xtgt$ %(target
%reached) 
or $V(\xi)\ge\lambda$. % (level crossed).

\begin{definition}[Reach-avoid contact certificate]\label{def:certificate}
A function $V$ satisfying (R1) is a \emph{reach-avoid contact
certificate} for $\Sigma$ at level $\lambda>0$ if for some $\zeta>0$:
\begin{align}
  V(x)\ &\ge\ 0, \quad \forall x\in\Xspace, \tag{C1}\label{c1}\\
  V(x)\ &\le\ 1, \quad \forall x\in\Xini, \tag{C2}\label{c2}\\
  V(x)\ &\ge\ \lambda, \quad \forall x\in\Xuns, \tag{C3}\label{c3}\\
  \gen V(x)\ &\le\ -\zeta, \quad \forall x\in \Flow\setminus\Xtgt
     \ \text{with}\ V(x)<\lambda, \tag{C4}\label{c4}\\
  \intv V(x^{-})\ &\le\ V(x^{-}), \quad \forall x^{-}\in\Gamma_\lambda.
     \tag{C5}\label{c5}
\end{align}
Conditions \eqref{c4}--\eqref{c5} are the flow and jump halves of a
quasi-variational inequality \cite{bensoussan1984impulse}:
$-\gen V\ge\zeta$ on the flow part of the corridor
$\{V<\lambda\}\setminus\Xtgt$ and $V\ge\intv V$ on its guard part
$\Gamma_\lambda$, with boundary values satisfying \eqref{c2}--\eqref{c3}.
\end{definition}

\begin{definition}[Compatible policy]\label{def:compat}
A reset policy $\pi$ (R3) is \emph{compatible} with $V$ if
\begin{equation}\label{eq:compat}
  \E_{Q}\big[V(x^{+})\,\big|\,x^{-},\pi(x^{-})\big]\ \le\ V(x^{-})
  \qquad\forall x^{-}\in\Gamma_\lambda .
\end{equation}
\end{definition}

Existence of a compatible policy implies \eqref{c5}, since
$\intv V\le\E_Q[V(x^{+})\mid\cdot\,,\pi(\cdot)]$ pointwise; conversely
\eqref{c5} asserts that adequate actions \emph{exist} at every point of $\Gamma_\lambda$, not that a measurable map selects them. The theorem is
therefore stated for a declared compatible $\pi$; the familiar
``greedy shield'' is recovered in
Corollary~\ref{cor:shield}.

\begin{theorem}[Reach-avoid guarantee]\label{thm:main}
If Asm.~\ref{asm:wellposed}--\ref{asm:regularity} hold, $V$ is a
reach-avoid contact certificate, and the reset policy $\pi$ is
compatible with $V$, then the closed loop system that applies $\pi$ at every
contact satisfies, for every $x_0\in\Xini$,
\begin{align}
 \Prob_{x_0}\big[\{\xi\in\Xtgt\}\wedge\{x_t\notin\Xuns,
     \forall t<\st\}\big]\ &\ge\ 1-\tfrac{1}{\lambda}\;=\;p,
     \label{eq:main-bound}\\
  \Prob_{x_0}[\st<\infty]\ &=\ 1. \label{eq:main-finite}
\end{align}
\end{theorem}

\begin{remark}[Pointwise form]\label{rem:pointwise}
In conditions \eqref{c1} and \eqref{c3}--\eqref{c5}, and compatibility,
%are invariant under replacing 
$(V,\lambda)$ can be replaced by $(V/c,\lambda/c)$ for
$c\in(0,\lambda)$, while \eqref{c2} then reads $V\le c$. Applying
Thm. \ref{thm:main} to the rescaled pair with $\Xini:=\{x_0\}$, and letting $c$ decrease to $V(x_0)$, gives, for every $x_0$ with
$V(x_0)<\lambda$,
$\Prob_{x_0}[\{\xi\in\Xtgt\}\wedge\{x_t\notin\Xuns,\ \forall t<\st\}]
\ge1-V(x_0)/\lambda$: the certificate is a pointwise bound, and
\eqref{c2} only normalizes it on $\Xini$.
\end{remark}

\begin{corollary}[Greedy shield]\label{cor:shield}
If for each $x^{-}\in\Gamma_\lambda$ the infimum in
\eqref{eq:intervention} is attained and
$u^{\star}(x^{-})\in\argmin_{u\in A(x^{-})}\E_{Q}[V(x^{+})\mid x^{-},u]$
defines a Borel map, then $u^{\star}$ is compatible with $V$ by
\eqref{c5}, and Theorem~\ref{thm:main} holds for the shielded loop
$u^{\star}$. More generally, if $b$ is \emph{any} Borel base proposal
and $\pi_c$ any compatible policy, the shield that keeps $b(x^{-})$
whenever $\E_Q[V(x^{+})\mid x^{-},b(x^{-})]\le V(x^{-})$ and falls back
to $\pi_c(x^{-})$ otherwise is again Borel and compatible: the
guarantee is inherited by arbitrary base policies at the price of
pointwise checks only.
\end{corollary}

The proof of Theorem~\ref{thm:main} rests on three lemmas: executions
are non-Zeno (Lemma~\ref{lem:non-zeno}); a suitably frozen copy of
$V(x_t)$ is a bounded supermartingale (Lemma~\ref{lem:frozen}); and the
corridor is left in finite expected time (Lemma~\ref{lem:settling}).

\begin{lemma}[Non-Zeno]\label{lem:non-zeno}
Under Assumption~\ref{asm:wellposed}, set
\[
  \eta:=\frac{\delta}{2\amax},\;\;
  p_0:=1-\exp\!\Big(-\frac{\delta\amax}{4\sbar^{2}}\Big),\;\;
  \tmin:=\eta\,p_0 ,
\]
with $\Delta t_k:=\infty$ on $\{\tau_k=\infty\}$ by convention, then for all $k$,
\begin{equation}\label{eq:nonzeno}
  \Prob[\Delta t_k\ge\eta\mid\Filt_{\tau_k}]\ \ge\ p_0
  \text{ and }
  \E[\Delta t_k\mid\Filt_{\tau_k}]\ \ge\ \tmin .
\end{equation}
Consequently, $\sum_k\Delta t_k=\infty$ almost surely.% so contacts do not
% accumulate infinitely in finite time.
\end{lemma}

\noindent Intuitively, Lemma 1 ensures that every finite time window contains finitely many
contacts. 

\begin{proof}
Define $s_t:=h(x_t)$, and abbreviate the quantity %separated 
in (B1) as
$\sep_t:=s_t$ in case (i) and $\sep_t:=\dot s_t=\gen h(x_t)$ in case
(ii), so that $\sep_{\tau_k^{+}}\ge\delta$ in both, enlarging the
constants, $\amax, \sbar>0$ of (B2) if necessary. %, $\amax,\sbar>0$. 
Fix $k$; on
$\{\tau_k=\infty\}$ both bounds hold trivially, so focus on
$\{\tau_k<\infty\}$ and argue conditionally on $\Filt_{\tau_k}$, under
which the excursion is the strong solution of \eqref{eq:shs} started at
$x^{+}_k$ and driven by the increments of $W$ after the stopping time
$\tau_k$. It is globally defined between contacts because $f,g$ are
Lipschitz with linear growth.

\emph{Step 1 (reduction to $\sep$).} The excursion cannot end before
$\sep$ first vanishes: in case (i), %because 
it ends exactly when
$s=0=\sep$; in case (ii), $\nabla h^{\top}g\equiv0$ removes the
martingale component of $s$, so $s$ is $C^{1}$ with $\dot s=\sep$ and,
starting from $s_{\tau_k^{+}}\ge0$, increases strictly while $\sep>0$.
Hence, with $\sigma_k:=\inf\{t>\tau_k:\sep_t=0\}$,
$\Delta t_k\ge\sigma_k-\tau_k$, and $\sep_t>0$ on $[\tau_k,\sigma_k)$.

\emph{Step 2 (uniform stochastic bound).} Since $h\in C^{4}$ and
$f,g\in C^{2}$, $\sep$ is a $C^{2}$ function of the state in both
cases, and It\^o's formula gives
$\mathrm{d}\sep_t=a(x_t)\,\mathrm{d}t+G(x_t)\,\mathrm{d}W_t$ along the
excursion, where $(a,G)=(\gen h,\ \nabla h^{\top}g)$ in case (i) and
$(\gen(\gen h),\ \nabla(\gen h)^{\top}g)$ in case (ii), so
$|a|\le\amax$ and $\|G\|\le\sbar$ by (B2). Let
\[
  M_t:=\int_0^{t\wedge(\sigma_k-\tau_k)}
        G(x_{\tau_k+r})\,\mathrm{d}W_r,\qquad t\ge0,
\]
the martingale part of $\sep$ stopped at $\sigma_k$. $M_t$ is a continuous
martingale for the filtration $(\Filt_{\tau_k+t})_{t\ge0}$ with
$\langle M\rangle_t\le\sbar^{2}t$. On $\{\sigma_k-\tau_k\le\eta\}$,
\begin{align*}
  0&=\sep_{\sigma_k}
 =\sep_{\tau_k^{+}}+\int_{\tau_k}^{\sigma_k}a(x_r)\,\mathrm{d}r
     +M_{\sigma_k-\tau_k}
   \\ & \ge\ \delta-\amax\eta+M_{\sigma_k-\tau_k}
   =\tfrac{\delta}{2}+M_{\sigma_k-\tau_k},
\end{align*}
so $\inf_{t\le\eta}M_t\le-\delta/2$ on that event. The exponential
inequality for continuous local martingales
\cite[Ch.~IV, Ex.~3.16]{revuz1999continuous}, applied to $-M$ with
$\langle M\rangle_\eta\le\sbar^{2}\eta$, yields
\begin{equation*}
\begin{aligned}
  \Prob[\sigma_k-\tau_k\le\eta\mid\Filt_{\tau_k}]
  \ & \le\ \exp\!\Big(-\frac{(\delta/2)^{2}}{2\sbar^{2}\eta}\Big)
  \\& =\ \exp\!\Big(-\frac{\delta\amax}{4\sbar^{2}}\Big)
  \ =\ 1-p_0 .
\end{aligned}
\end{equation*}
Since $\{\sigma_k-\tau_k>\eta\}\subseteq\{\Delta t_k\ge\eta\}$ by
Step~1, the first bound in \eqref{eq:nonzeno} follows.

\emph{Step 3 (expectation bound).} By Markov's inequality,
$\E[\Delta t_k\mid\Filt_{\tau_k}]
 \ge\eta\,\Prob[\Delta t_k\ge\eta\mid\Filt_{\tau_k}]\ge\eta p_0=\tmin$.

\emph{Step 4 (Zeno excluded).} If only finitely many contacts occur,
the last excursion never ends and some $\Delta t_k=\infty$. Otherwise
$\tau_k<\infty$ for all $k$, the events $A_k:=\{\Delta t_k\ge\eta\}$
satisfy $A_k\in\Filt_{\tau_{k+1}}$ and
$\sum_k\Prob[A_k\mid\Filt_{\tau_k}]\ge\sum_k p_0=\infty$, so the
conditional Borel--Cantelli lemma
\cite[Sec.~12.15]{williams1991probability} gives $A_k$ infinitely often
almost surely, hence $\sum_k\Delta t_k=\infty$.
\end{proof}

Define the frozen process
\begin{equation}\label{eq:frozen}
  \Yproc_t:=
  \begin{cases}
    V(x_t), & t<\st,\\
    0, & t\ge\st \text{ and } \xi\in\Xtgt,\\
    \lambda, & \text{otherwise:}
  \end{cases}
\end{equation}
$Y$ follows $V$ until the corridor is left, then freezes at $0$ (target
reached first) or at $\lambda$ (level crossed first).

\begin{lemma}[Frozen supermartingale]\label{lem:frozen}
Under Asm. ~\ref{asm:wellposed}--\ref{asm:regularity}, if $V$
satisfies \eqref{c1}, \eqref{c4} and $\pi$ is compatible with $V$, 
then
$(\Yproc_t)_{t\ge0}$ is a nonnegative $(\Filt_t)$-supermartingale with
$\Yproc_t\in[0,\lambda]$  $\forall t$.
\end{lemma}

\begin{proof}
$\st$ is a stopping time (debut theorem), and
$\Yproc_t\in[0,\lambda]$ by \eqref{c1} and the definition of $\st$; this
boundedness is the sole source of integrability used below. On a
jump-free arc $[a,b]$ of an excursion with $b<\st$, It\^o's formula
gives $V(x_b)=V(x_a)+\int_a^b\gen V(x_r)\,\mathrm{d}r+N_b$ with
$N_t:=\int_a^t\nabla V^{\top}g\,\mathrm{d}W_r$. Along the arc
$V\in[0,\lambda)$ and $\gen V\le-\zeta<0$ by \eqref{c4}, so
$N_t\ge-\lambda$; a local martingale bounded below is a supermartingale,
hence $\E[N_b|\Filt_a]\le0$ and $\E[V(x_b)|\Filt_a]\le V(x_a)$.
At a contact $\tau\le\st$, if $x^{-}$ itself enters
$\Xtgt\cup\{V\ge\lambda\}$ then $\tau=\st$ and the frozen value
$\Yproc_\tau\in\{0,\lambda\}$ satisfies
$\Yproc_\tau\le V(x^{-})=\Yproc_{\tau^{-}}$ trivially. Otherwise $x^{-}\in\Gamma_\lambda$, and in each post-jump case
($V(x^{+})<\lambda$ survival; $x^{+}\in\Xtgt$; overshoot)
$\Yproc_\tau\le V(x^{+})$ pointwise, so, the reset randomness being
independent of $\Filt_{\tau^{-}}$,
\[
  \E[\Yproc_\tau\mid\Filt_{\tau^{-}}]
  \le \E_{Q}\big[V(x^{+})\mid x^{-},\pi(x^{-})\big]
  \le V(x^{-})=\Yproc_{\tau^{-}}
\]
by \eqref{eq:compat}. By Lemma~\ref{lem:non-zeno}, $[s,t]$ contains
finitely many contacts a.s. Applying the two estimates alternately,
through the tower property, to the process stopped at the $n$-th
contact gives $\E[\Yproc_{t\wedge\tau_n\wedge\st}\mid\Filt_s]\le\Yproc_s$
for every $n$; letting $n\to\infty$ (bounded convergence,
$\Yproc\in[0,\lambda]$) yields
$\E[\Yproc_{t\wedge\st}\mid\Filt_s]\le\Yproc_s$ and, since $\Yproc$ is
constant after $\st$, the claim follows.
\end{proof}

\begin{lemma}[Finite settling]\label{lem:settling}
Under Asm.~\ref{asm:wellposed}--\ref{asm:regularity}, if $V$ satisfies \eqref{c1}, \eqref{c4} and $\pi$ is compatible with $V$, then
\begin{equation}\label{eq:settling}
  \E[\st]\ \le\ \frac{V(x_0)}{\zeta}\ <\ \infty,
  \;\;\text{hence}\;\;
  \Prob[\st<\infty]=1 .
\end{equation}
\end{lemma}

\begin{proof}
Let $M_t:=\Yproc_t+\zeta\,(t\wedge\st)$, a process with values in
$[0,\lambda+\zeta t]$. On flow arcs inside the corridor, \eqref{c4}
gives $M_b-M_a=\int_a^b(\gen V+\zeta)\,\mathrm{d}r+N_b\le N_b$ with $N$
the stochastic integral of Lem. \ref{lem:frozen}, 
%the previous proof, 
so
$\E[M_b-M_a\mid\Filt_a]\le0$; at a contact the clock term
is unchanged
and $\Yproc$ does not increase in conditional expectation. %, as shown
% there. 
Splicing as in Lem.~\ref{lem:frozen} gives
$\E[M_{t}]\le M_0\le V(x_0)$, that is,
$\zeta\,\E[t\wedge\st]\le V(x_0)-\E[\Yproc_t]\le V(x_0)$ by \eqref{c1};
monotone convergence as $t\to\infty$ gives \eqref{eq:settling}, and a
nonnegative variable with finite mean is finite a.s.
\end{proof}

\begin{corollary}[Contact budget]\label{cor:budget}
Under the hypotheses of Lemma~\ref{lem:settling}, with $\eta,p_0$ as
in Lemma~\ref{lem:non-zeno} and $\jumps$ the number of contacts before
$\st$,
\[
  \E[\jumps]\ \le\ \frac{\E[\st]}{p_0\eta}+\frac1{p_0}
  \ \le\ \frac{2\amax V(x_0)}{\zeta\delta p_0}+\frac1{p_0}.
\]
Moreover, if \eqref{eq:compat} holds with a margin $\zeta_d > 0 $
($\E_Q[V(x^{+})\mid x^{-},\pi(x^{-})]\le V(x^{-})-\zeta_d$ on
$\Gamma_\lambda$), then 
$\E[\jumps]\le V(x_0)/\zeta_d$.
\end{corollary}

\begin{proof}
Pathwise, at most $\st/\eta+1$ of the dwell intervals preceding $\st$
have $\Delta t_k\ge\eta$ (the last may straddle $\st$); the
expected number of the remaining dwell intervals is at most $(1-p_0)\,\E[\jumps]$
by \eqref{eq:nonzeno}, since $\ind{\tau_k<\st}$ is
$\Filt_{\tau_k}$-measurable. Hence,
$\E[\jumps]\le\E[\st]/\eta+1+(1-p_0)\E[\jumps]$; applied to
$\jumps\wedge n$ and letting $n\to\infty$, this gives the first bound,
and $\eta=\delta/2\amax$ with \eqref{eq:settling} the second. For the
last claim, $\Yproc_t+\zeta_d\,j_t$, with $j_t$ the number of contacts
in $(0,t\wedge\st)$, is a supermartingale as in
Lem.~\ref{lem:settling} with the clock replaced by contact count,
so $\zeta_d\,\E[j_t]\le V(x_0)$; let $t\to\infty$.
\end{proof}

\subsection{Proof of the Main Theorem}\label{sec:mainproof}

\begin{proof}[Proof of Theorem~\ref{thm:main}]
\emph{Step 1 (level crossing).} By Lemma~\ref{lem:frozen}, $(\Yproc_t)$
is a nonnegative supermartingale with
$\E[\Yproc_0]\le V(x_0)\le1$ by \eqref{c2}, %($\Yproc_0=0$ if $x_0\in\Xtgt$ and $\Yproc_0=V(x_0)$ otherwise), 
so Ville's maximal inequality
gives
\begin{equation}\label{eq:ville}
  \Prob_{x_0}\Big[\sup_{t\ge0}\Yproc_t\ge\lambda\Big]
  \ \le\ \frac{\E[\Yproc_0]}{\lambda}\ \le\ \frac1\lambda .
\end{equation}
\emph{Step 2 (pre-$\st$ safety is sure).} If $x_t\in\Xuns$ for some
$t<\st$ then $V(x_t)\ge\lambda$ by \eqref{c3}, contradicting $t<\st$;
hence $\{x_t\notin\Xuns,\ \forall t<\st\}$ holds $\Prob_{x_0}$-a.s.
\emph{Step 3 (reach).} By Lem.~\ref{lem:settling}, $\st<\infty$ a.s.,
and the entering state $\xi$ is well defined:  $\Xtgt\cup\{V\ge\lambda\}$
is closed, so an entrance along a flow arc is attained, while at a
contact one of $x^{-},x^{+}$ enters by definition. On
$\{\xi\notin\Xtgt\}$ one has $V(\xi)\ge\lambda$, hence $\Yproc$ freezes
at $\lambda$ and
$\{\xi\notin\Xtgt\}\subseteq\{\sup_t\Yproc_t\ge\lambda\}$; by
\eqref{eq:ville}, $\Prob_{x_0}[\xi\in\Xtgt]\ge1-1/\lambda$.
\emph{Step 4 (conclusion).} Intersecting the reach event of Step~3 with
the sure avoid event of Step~2 leaves the probability unchanged, which
is \eqref{eq:main-bound}; \eqref{eq:main-finite} is
Lemma~\ref{lem:settling}.
\end{proof}

\begin{remark}[Compatibility on an invariant guard section]
\label{rem:kpi}
The proof invokes \eqref{eq:compat} only at $x^{-}_k$. 
Hence \eqref{eq:compat} (and \eqref{c5}) may be
imposed on $\Gamma_\lambda\cap G$ only, where $G\subseteq\Guard$ is
a Borel guard section for which one additionally proves that the
contact chain never leaves $G$ before $\st$, a.s.;
Theorem~\ref{thm:main} is then unchanged. If the chain can leave $G$
with probability at most $\varepsilon$, stopping at the first exit
and counting it as a failure weakens \eqref{eq:main-bound} to
$1-1/\lambda-\varepsilon$. The same localization applies to the flow
condition: if \eqref{c4} is verified only on the corridor intersected
with a closed set $K\subseteq\Xspace$, stop the execution at the first
exit from $K$ and count it as a failure; then \eqref{eq:main-bound}
holds with its right side lowered by
$\Prob_{x_0}[\text{exit from }K\text{ before }\st]$.
\end{remark}

\section{Experimental Results}\label{sec:results}

Certificates are approximately computable because every condition of
Def.~\ref{def:certificate} is linear in $V$. Fix a feature map
$\varphi=(\varphi_1,\dots,\varphi_m)^{\top}$ with each
$\varphi_j\in C^{2}$, write $V(x)=\varphi(x)^{\top}\theta$ with
unknown coefficients $\theta\in\R^{m}$, and place finite grids
$\hat{\mathcal X}\subset\Xspace$, $\hat{\mathcal X}_{0}\subset\Xini$,
$\hat{\mathcal X}_{u}\subset\Xuns$, and $\hat{\mathcal X}_{F}$,
$\hat\Gamma$ in the domains of \eqref{c4} and \eqref{c5}. For a %prescribed 
rate $\zeta>0$ and a
declared policy $\pi$, a certificate search is performed via the linear program
\begin{align}
\min_{\theta\in\R^{m}}\ & J(\theta) \label{eq:lp} \\ \text{s.t.}\ \
&\varphi(x)^{\top}\theta\ge0,\ x\in\hat{\mathcal X};\quad
 \varphi(x)^{\top}\theta\le1,\ x\in\hat{\mathcal X}_{0}; \nonumber \\
&\varphi(x)^{\top}\theta\ge\lambda,\ x\in\hat{\mathcal X}_{u};\
 (\gen\varphi)(x)^{\top}\theta\le-\zeta,\ x\in\hat{\mathcal X}_{F}; \nonumber \\
&\big(\E_{Q}[\varphi(x^{+})\mid x,\pi(x)]-\varphi(x)\big)^{\top}
 \theta\le0,\ x\in\hat\Gamma, \nonumber
\end{align}
with some LP-representable objective $J$.The rows of \eqref{eq:lp} require only the numbers $\varphi(x)$, $(\gen\varphi)(x)$ and $\E_{Q}[\varphi(x^{+})]$, so any $C^{2}$
dictionary and any reset kernel with computable feature means
qualify. A feasible $\theta$, however, is a candidate rather than a
certificate, because its constraints hold only at the grid points.
Correctness therefore rests on an independent verifier: it encloses $V$, $\gen V$ and $\E_Q[V(x^{+})]$ on boxes covering the corridor and the guard (Bernstein coefficients for polynomial blocks, interval
extensions for the remaining features) and bisects until every
condition is decided with a positive margin; boxes that resist are
reported. Any dictionary of interval-computable
features qualifies, and the number of boxes grows with the state dimension, not degree. Violations return to \eqref{eq:lp} as new constraints. 

The examples verify this pipeline in three ways. Ex.~1 solves
\eqref{eq:lp} over a feature class adapted to the operator and
measures the conservatism against an exact PDE ground
truth. Ex.~2 solves it over a generic polynomial basis and
re-proves the verified candidate in exact rational arithmetic.
Ex.~3 performs no synthesis: its certificate is closed form, the
verifier proves its flow condition, and its jump condition is verified
by quadrature on a grid. When dynamics, sets, and features are all polynomial, \eqref{c1}--\eqref{c5} can
instead be imposed exactly by sum-of-squares programming
\cite{prajna2007framework,santoyo2021barrier}, without grids or a
verifier; the grid LP gives up between-grid exactness in exchange for
a larger dictionary, on which Ex.~1 (logarithms, multipoles) and
Ex.~3 (square roots) rely.

\subsection{Example 1: Escape from a Debris Field}\label{sec:corridor}

\emph{System.} A probe drifts across a debris field
$\Xspace=[-3,3]\times[-2.2,2.2]$ under a weak current and a dominant
diffusion,
\begin{equation}\label{eq:corr-sde}
  \mathrm dx=f\,\mathrm dt+\sigma\,\mathrm dW,\qquad
  f=(0.12,\,0)^{\top},\quad \sigma=0.55.
\end{equation}
The task is to reach the outermost strip, of width $0.01$, along the right, top, or bottom edge without touching any of three debris discs (Fig.~\ref{fig:hopper}):
% \begin{equation}\label{eq:corr-sets}
  $\Xtgt=\{x_1\ge2.99\}\cup\{|x_2|\ge2.19\},
\Xuns=\textstyle\bigcup_{i=1}^{3}B(c_i,r_i).$
% \end{equation}
The left edge is the guard, $\Guard=\{x_1=-3\}\setminus\Xtgt$.
A probe swept out upstream is re-inserted by a recovery system whose
kernel is
\begin{equation}\label{eq:corr-reset}
  Q(\cdot\mid x^{-})
  =0.9\;\mathrm{Unif}\big(\{2.05\}\times[-1.8,1.8]\big)
  +0.1\,\delta_{c_1},
\end{equation}
so recovery fails with probability $0.1$ and leaves the probe at the
center of the first disc, inside $\Xuns$. Assumption~\ref{asm:wellposed}
holds in case (i) with $h=x_1+3$: both branches of \eqref{eq:corr-reset}
land at $h\ge\delta=2.05$, and $\gen h=f_1$,
$\|\nabla h^{\top}g\|=\sigma$ give (B2). 

\emph{Ground truth.} We take $\lambda=10$ and claim level
$\bar p=0.8$. Theorem~\ref{thm:main} in its pointwise form (Remark~\ref{rem:pointwise}) certifies the region
$S_{\rm cert}:=\{V\le\lambda(1-\bar p)\}=\{V\le2\}$, and soundness
places it inside $S_{\rm true}:=\{p\ge\bar p\}$, where $p(x)$ is the
reach-avoid probability from $x$. Conservatism is then the area
between the two boundaries, and it is computable here because $p$
is. This $p$ and the mean settling time $T$ satisfy $\gen p=0$ and
$\gen T=-1$, with $p=0$ on $\Xuns$, $p=1$ on $\Xtgt$, $T=0$ on both
\cite{oksendal2003sde}, and with the nonlocal guard condition
$p|_{\Guard}=\E[p(x^{+})]$, likewise for $T$. The guard value is a
single scalar $q$, since $Q$ does not depend on $x^{-}$, and the
interior solution is affine in $q$; two sparse solves on a
$421\times311$ grid therefore give $\E[p(x^{+})]$ as an affine
function of $q$, and the guard condition becomes one linear equation
for $q$.

\emph{Synthesis and verification.} We solve \eqref{eq:lp} at rate
$\zeta=10^{-2}$ with $m=211$ features: a Chebyshev basis of total
degree $18$, augmented at
each disc by the harmonic multipoles $\log|w|$, $\mathrm{Re}\,w^{-n}$,
$\mathrm{Im}\,w^{-n}$ ($w=z-c_i$, $n\le3$). The multipoles adapt the
dictionary to the operator: harmonic features carry no diffusion
term, $\gen\varphi=f\cdot\nabla\varphi$, and without them the LP meets
\eqref{c4} at its grid points while $\gen V$ reaches $+19$ between
them. The objective $J=\sum_i\max\{0,V(x_i)-2\}$ penalizes the excess
of $V$ over the claim level, which enlarges $S_{\rm cert}$. The
verifier then proves the candidate: \eqref{c1} and \eqref{c4} by
interval branch and bound over the state box ($6.9\times10^{5}$
boxes), the latter at the certified rate $\zeta_v=4.4\times10^{-3}$
that enters Theorem~\ref{thm:main}, and \eqref{c5} in exact rational
arithmetic on the guard ($5.7\times10^{4}$ boxes). It found
two defects missed by grid checks, $V$ dipping to
$-5\times10^{-4}$ and below $\lambda$ near the disc centers, and the
candidate was shifted accordingly by
$\varepsilon=5\times10^{-4}$, at a cost of
$\varepsilon/\lambda=5\times10^{-5}$ in the bound, and raised to
$\lambda$ inside the discs, where it is never differentiated and
where the only reset landing point already has $V\ge\lambda$.

\emph{Results.} Fig.~\ref{fig:hopper}(a) overlays the certified
boundary $\{V=2\}$ on the true boundary $\{p=\bar p\}$: the certified
region lies inside the true one, as Theorem~\ref{thm:main} requires,
and covers $93.6\%$ of its area. Fig.~\ref{fig:hopper}(b) plots the
pointwise deficit $p-(1-V/\lambda)$, nonnegative at every state, with
median $2.15$ percentage points (pp) overall and $0.60$\,pp where
$p\ge0.9$. To attribute the gap, we compare with the canonical
certificate at the LP's rate, $V^{\star}:=\lambda (1-p)+\zeta T$, which
satisfies \eqref{c1} and \eqref{c3}--\eqref{c5}, the last with
equality. The delivered $V$ dominates $V^{\star}$ pointwise (checked
on the grid), so $S_{\rm cert}\subseteq\{V^{\star}\le2\}$; the latter
covers $99.2\%$ of the true area, and the term $\zeta T/\lambda$ has
median $0.21$\,pp. Reading the same certificate at other claim levels
$\bar p\in[0.6,0.95]$ moves the coverage of
$\{V\le\lambda(1-\bar p)\}$ from $95.2\%$ to $89.5\%$, while that of
$\{V^{\star}\le\lambda(1-\bar p)\}$ stays within $98.4$--$99.5\%$. The
conservatism is therefore a property of the function class, not of
Theorem~\ref{thm:main}. Along the guard \eqref{c5} is tight (margin
$2\times10^{-5}$), and the failed-recovery branch contributes $37\%$
of $\E[V(x^{+})]$.

\begin{figure*}[t]
  \centering
  \vspace{0.01in}
  \includegraphics[width=0.88\textwidth]{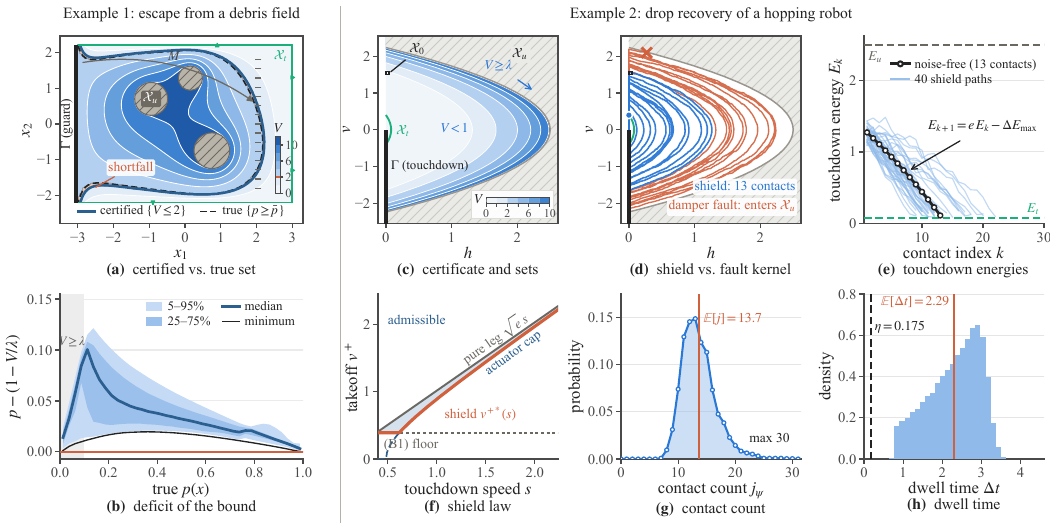}
\caption{\small \textbf{Example 1.} (a)~Certificate field $V$:
the certified boundary $\{V=\lambda(1-\bar p)\}$ against
the true boundary $\{p=\bar p\}$ from the PDE; the coral
strip between them is the conservatism, $6.4\%$ by area, and the rug
$M$ marks the support of the reset kernel. (b)~Pointwise deficit
$p-(1-V/\lambda)$ against $p$, with median, percentile bands and
per-bin minimum: nonnegative everywhere, median $2.15$\,pp, median
$0.60$\,pp where $p\ge0.9$; in the shaded band $V\ge\lambda$ and no
claim is made. \textbf{Example 2.} (c)~The degree-8
certificate with the target $\Xtgt$, the unsafe set $\Xuns$ (hatched),
the guard $\Gamma$ and the initial set $\Xini$. (d)~One shielded
recovery against one path under the kernel whose damper fails at every
stance. (e)~Touchdown energies of $40$ shielded paths against the
noise-free recursion $E_{k+1}=eE_k-\Delta E_{\max}$. (f)~The shield law
$v^{+*}(s)$ inside the admissible takeoff set, bounded by the actuator
cap, the pure-leg line and the (B1) floor. (g)~Distribution of the
contact count. (h)~Distribution of the dwell time between contacts.\vspace{-1.7em}}
  \label{fig:hopper}
\end{figure*}

\subsection{Example 2: Drop Recovery of a Hopping Robot}\label{sec:hopper}

\emph{System.} A Raibert-style monoped lands with too much energy
and must dissipate the excess over many stances
\cite{raibert1986legged}. Flight is exactly ballistic with velocity
noise,
\begin{equation}\label{eq:hop-flow}
  \mathrm dh=v\,\mathrm dt,\qquad
  \mathrm dv=-g\,\mathrm dt+\sigma\,\mathrm dW,\qquad
  E:=\tfrac12v^{2}+gh, \nonumber
\end{equation}
with $g=1$ and $\sigma=0.05$; the specification for 
energy $E$ is:
  $\Xtgt=\{E\le0.08\},\;
  \Xuns=\{E\ge2.5\}, \; \textup{and }
  \Xini=[0,0.05]\times[1.50,1.58]$
the over-energetic launch
($E_0\in[1.13,1.30]$). 
At touchdown, on $\Guard=\{h=0\}$ (where
$v^{-}\le0$, every reset relaunching upward), the stance
resets the velocity to
\begin{equation}\label{eq:hop-reset}
\begin{aligned}
  v^{+}&=\sqrt{e\,(v^{-})^{2}-2\Delta E}+\epsilon,\\
  \Delta E&\in\big[0, \min\{\Delta E_{\max},\,
  eE^{-}-\tfrac12(v_{\min}+w)^{2}\}\big],
\end{aligned}
\end{equation}
where the leg gain is $e=1.03$, the damper absorbs at most $\Delta E_{\max}=0.12$,
and $\epsilon\sim U[-w,w]$ with $w=0.04$. This is case (ii) of (B1),
with $h$ the height and $\gen h=v$: the noise enters $v$ only, so
$\nabla h^{\top}g\equiv0$, and the upper limit of the action set is
the truncation, keeping $v^{+}\ge\delta=v_{\min}=0.35$ almost surely at every contact outside the target ($|v^{-}|\ge0.40$), the only
contacts that precede $\st$; (B2) holds with $\amax=1$ (the gravitational acceleration) and $\sbar=\sigma$. With probability
$\rho_s=0.003$ the damper fails and the stance is pure leg
($\Delta E=0$), which increases the energy.

Because $\gen E=\sigma^{2}/2>0$, the flow never removes energy and
all progress is made at contacts. The additive cap turns the
touchdown energies into a staircase $E_{k+1}=eE_k-\Delta E$ that
needs about $13$ steps. We certify $p=0.90$ ($\lambda=10$).

\emph{Synthesis and verification.} We solve \eqref{eq:lp} over the
generic dictionary of all $m=45$ monomials $h^{i}v^{j}$ of total
degree at most $8$, for the declared policy that absorbs as much
energy as the action set allows,
$\Delta E=\min\{\Delta E_{\max},\,eE^{-}-\tfrac12(v_{\min}+w)^{2}\}$.
Under this policy the moments $\E[(v^{+})^{k}]$ of the reset law, a
mixture of the commanded takeoff and the damper fault, have closed
forms, so the rows of \eqref{c5} are exact rather than sampled. The
objective is feasibility: the LP maximizes a uniform margin over all
constraints at a fixed rate $\zeta$, and on a descending ladder of
rates it is first feasible at $\zeta=6\times10^{-3}$; degrees $4$ and
$6$ are infeasible at that rate. Low rates are intrinsic to the
example, since Lemma~\ref{lem:settling} gives
$\zeta\le V(x_0)/\E[\st]\le1/\E[\st]$ and the recovery takes about
$30$ time units. Because \eqref{c4} and \eqref{c5} are owed only on
the corridor $\{V<\lambda\}$, which depends on $V$, their grids are
scoped to the corridor of the previous iterate, and the points that a
finer grid finds violated return to the LP as rows; two iterations
suffice.

The verifier then proves the candidate, rounded to nine decimals, in
exact rational arithmetic: \eqref{c1}--\eqref{c4} by Bernstein branch
and bound over the state box, \eqref{c4} at the certified rate
$\zeta_v=5\times10^{-3}$, and \eqref{c5} along the guard section
outside the target, with the square roots of the reset law enclosed
by outward rational bounds ($380$ boxes in all). It refuted the first
candidate: on the boundary $\{E=E_u\}$ near the guard corner
($h\approx0$, $v\approx-2.24$) it found $V=\lambda-5\times10^{-3}$, a
violation of \eqref{c3} that both grids had stepped over. Rows
sampled along that curve were added to the LP, and the proof of the
new candidate closes. Along the guard, $V(0,\cdot)$ is increasing in the takeoff speed, so
the greedy shield of Corollary~\ref{cor:shield} needs no optimization: it commands the lowest admissible takeoff speed,
$v^{+*}(s)=\max\{\sqrt{e s^{2}-2\Delta E_{\max}},\,v_{\min}+w\}$ for a
touchdown at speed $s=|v^{-}|$, the actuator cap binding at fast
touchdowns and the (B1) floor $v_{\min}+w=0.39$ at slow ones. 

\emph{Results.} Over $6000$ closed-loop paths from the worst corner
of $\Xini$, the greedy shielded loop reaches $\Xtgt$ with
$\hat p=0.9997\ge0.90$ and no failures; the two paths cut by the
simulation horizon count against $\hat p$.
Fig.~\ref{fig:hopper}(c) shows the certificate: $V$ reaches
$\lambda$ at the structural limit, so \eqref{c3} alone confines the
corridor. Panel~(d) contrasts one shielded recovery with the kernel
whose damper fails at every stance: the shielded flight arcs shrink
toward $\Xtgt$, the faulty ones grow into $\Xuns$. Panel~(e) shows
the touchdown energies falling in near-constant decrements of at
most $\Delta E_{\max}$, so the actuator, not the noise, sets the
length of the staircase. Panel~(f) shows the shield law on the lower
edge of the action set, the cap binding at fast touchdowns and the
(B1) floor at slow ones. Panels~(g) and (h) test the budgets: the
contact count has mode $12$ and maximum $30$, with $\E[\jumps]=13.8$
and $\E[\st]=32.6$ against the bounds $1144$ and $200$ of
Corollary~\ref{cor:budget} and Lemma~\ref{lem:settling} at
$\zeta_v$, and over $n=76\,658$ contacts the mean dwell time is
$2.30$ against the $\eta=0.175$ of Lemma~\ref{lem:non-zeno}.

\subsection{Example 3: Dribbling through an Obstacle Course}\label{sec:dribble}

\emph{System.} A ball is dribbled $16$\,m across a field with three
staggered obstacles by kicking it whenever it comes to rest. Between
kicks it rolls under Coulomb friction and turf noise,
\begin{equation}\label{eq:dri-flow}
  \mathrm dp=v\,\mathrm dt,\qquad
  \mathrm dv=-A_R\,\tfrac{v}{|v|}\,\mathrm dt+\sigma_r\,\mathrm dW,
\end{equation}
with $A_R=0.5$ and $\sigma_r=5\times10^{-4}$. The specification is
  $\Xtgt=\{|p-p^{*}|\le0.5\},\;
  \Xuns=\textstyle\bigcup_{i=1}^{3}B(c_i,r_i), \;
  \Xini=\{|p|_\infty\le0.1,\ |v|=v_g\},$
with $p^{*}=(16,0)$ and obstacles centered at $(4.5,0.68)$,
$(8.5,-0.85)$, $(12.5,0.60)$ with radii $0.45$, $0.60$, $0.45$, so
that the straight line from the origin to the target passes the
nearest one at $0.15$\,m (Fig.~\ref{fig:dribble}). The guard is
$\Guard=\{|v|=v_g\}$ with $v_g=0.05$ and $h=|v|-v_g$; states on the guard are called rest states below. Every
episode starts on the guard near the origin, so the first kick is
taken at once. A kick chooses a speed and a direction,
$u=(s,\varphi)\in[0.7,1.8]\times[0,2\pi)$, and resets
\begin{equation}\label{eq:dri-reset}
  p^{+}=p,\qquad
  v^{+}=(s+k_s w)\,
  \big(\cos(\varphi+\theta),\,\sin(\varphi+\theta)\big)^{\top}, \nonumber
\end{equation}
with $w\sim U[-1,1]$, $k_s=0.1$, and aiming noise
$\theta\sim\mathcal N(0,\sigma_\theta^{2})$, $\sigma_\theta=0.05$,
clamped to $|\theta|\le0.1$. Asm.~\ref{asm:wellposed} holds in
case (i): $|v^{+}|\ge0.6$, so $\delta=0.6-v_g=0.55$,
and $\gen h=-A_R$ $+\sigma_r^{2}/(2|v|)$, $\|\nabla h^{\top}g\|=\sigma_r$
give (B2) with $\amax=0.5$, $\sbar=5\times10^{-4}$. We take
$\lambda=20$ and certify $p=0.95$.

\begin{figure*}[t]
  \centering
  \captionsetup{skip=1pt}
  \vspace{0.01in}
  \includegraphics[width=\textwidth]
  {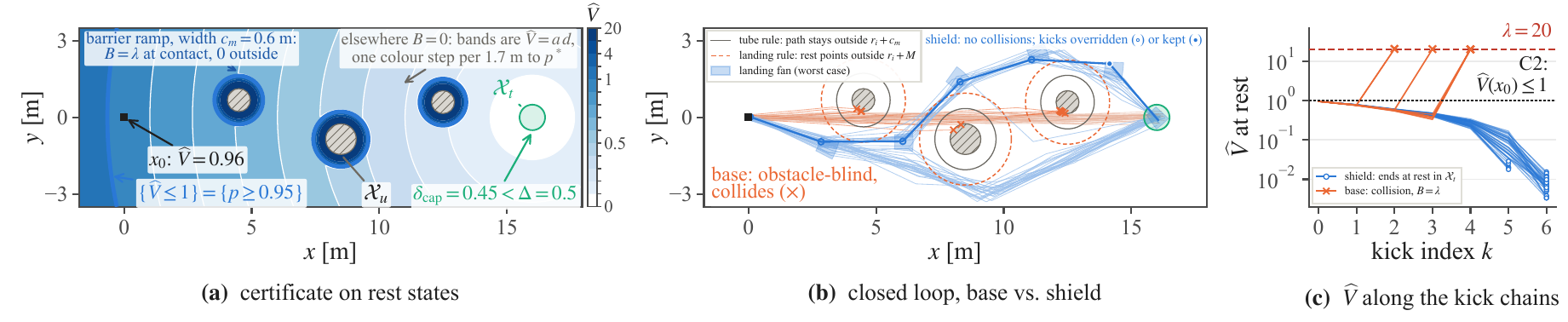}
  \caption{\small Example 3. (a)~The certificate $\widehat V$ on rest states;
    the level $\widehat V=1$ bounds the certified set $\{p\ge0.95\}$.
    (b)~Closed loop, $40$ episodes per policy: the obstacle-blind base
    against the shielded loop, with the worst-case landing fan of
    every kick of one shielded episode. (c)~$\widehat V$ along the
    kick chains, against the budget $\widehat V(x_0)\le1$ of
    \eqref{c2} and the level $\lambda$.
    \vspace{-1.8em}
    }
  \label{fig:dribble}
\end{figure*}

\emph{Certificate and verification.} Let
$p_{\rm stop}=p+v|v|/(2A_R)$ be where the ball would come
to rest under the drift alone, $c_i$ the clearance of the remaining
path $[p,p_{\rm stop}]$ from obstacle $i$, and $t_{\rm stop}$ the time
to rest, mollified at $v=0$. We take
\begin{equation}\label{eq:dri-cert}
  \widehat V=a\,D_\varepsilon+\textstyle\sum_{i}B(c_i)
  +\kappa\,t_{\rm stop},
\end{equation}
with $D_\varepsilon$ the mollified distance from $p_{\rm stop}$ to
$p^{*}$ and $B(c)=\lambda(1-u)^{3}(1+3u)$, $u=c/c_m$ clipped to
$[0,1]$, a ramp from $B=\lambda$ at contact to $B=0$ at the barrier
width $c_m$. The parameters follow from the conditions: $c_m=0.6$\,m
so that the ramps fit the gaps between the obstacles, $a=0.06$ so that
\eqref{c2} holds, with $\widehat V\le0.967$ on $\Xini$, and
$\kappa=6.8\times10^{-3}$, $15\%$ above the smallest value at which
the drift dominates the noise in \eqref{c4}. Conditions \eqref{c1}
and \eqref{c3} hold by construction, since every term is nonnegative
and $B=\lambda$ on $\Xuns$. For \eqref{c4}, under the drift alone
$p_{\rm stop}$ is constant, each $c_i$ is nondecreasing, and
$t_{\rm stop}$ decreases at a rate of at least $0.7$ on the band
$v_g\le|v|\le\bar v=1.95$, so the drift of $\gen\widehat V$ is at
most $-0.7\kappa$; the It\^o term is enclosed by the verifier on
$5.8\times10^{5}$ boxes in $(p,|v|,\angle v)$, which gives
$\gen\widehat V\le-\zeta_v$ with $\zeta_v=2.2\times10^{-3}$ on the
band for every $p$. Kicks land at $|v^{+}|\le1.9<\bar v$, and by
Remark~\ref{rem:kpi} the probability of leaving the band before $\st$
is at most $\E[\jumps]\exp(-2A_R(\bar v-1.9)/\sigma_r^{2})$,
numerically zero. Condition \eqref{c5} is verified on a grid of rest states, a Lipschitz
correction extending each verdict to its cell: for every admissible
kick, the barrier term vanishes on the whole support of the kick noise,
by the tube rule, so the integrand of $\E[\widehat V(x^{+})\mid x,u]$
is smooth and its quadrature error is below $2\times10^{-5}$; the
condition holds with margin $\zeta_d=5\times10^{-3}$ at every rest
state outside $\Xtgt$ with clearance at least $1.09$\,m. Finally, $\widehat V$ is $C^{1}$ with locally Lipschitz gradient and, on the corridor, $C^{2}$ off the surfaces where the nearest point of the predicted path to an obstacle switches between an endpoint and the interior of the segment, which the velocity noise crosses transversally, so (R1) holds in the sense of Remark~\ref{rem:asm-roles}.

\emph{Policy.} The reset policy is the shield of
Corollary~\ref{cor:shield} around a base proposal that aims at the
target with range-matched speed. The proposal is kept if the swept
path of the kick clears every obstacle by $c_m$ (tube rule), its
worst-case stop point clears every obstacle by $M=1.2$\,m (landing
rule), and the expected decrease of $\widehat V$ is at least
$\zeta_d$; otherwise the policy plays the
kick with the largest expected decrease among those on a fixed grid
of four speeds and $25$ deflections that pass the first two tests.
The landing rule keeps every rest state of the chain at clearance
$M$ or more, inside the region where \eqref{c5} holds, which
localizes compatibility as in Remark~\ref{rem:kpi}. With
$\widehat V\le0.967$ on $\Xini$, Remark~\ref{rem:pointwise} certifies
$p\ge1-0.967/20=0.951$, Corollary~\ref{cor:budget} gives
$\E[\jumps]\le0.967/\zeta_d=194$ kicks, and Lemma~\ref{lem:settling}
gives $\E[\st]\le0.967/\zeta_v=440$\,s.

\emph{Results.} Fig.~\ref{fig:dribble}(a) shows $\widehat V$ on rest states. Its level set $\widehat V=1=\lambda(1-0.95)$ bounds the certified set $\{p\ge0.95\}$: a disc of radius $16.7$\,m about the target with a hole punched at each obstacle, and a certified capture radius of $0.45$\,m inside the $0.5$\,m task disc, so the guarantee is
end-to-end. Panel~(b) shows the closed loop. Over $200$ episodes the
obstacle-blind base policy collides in $68$; the shielded loop reaches
the target in all $200$ without a collision, overriding $60\%$ of its
$5.9$ kicks per episode, and every one of its $1174$ rest states
admits a certified kick. The shielded paths pass each obstacle on the
side away from the straight line, their swept paths outside the solid
rings and their rest points outside the dashed ones. Panel~(c) tracks $\widehat V$ along the kick chains: shielded chains
start at $\widehat V(x_0)=0.96$, within the budget of \eqref{c2}, and
decrease to the target, while a base collision drives $B$, and with it
$\widehat V$, to $\lambda$, so that a collision is a level crossing in
the sense of \eqref{eq:ville}.
All $233$ kicks played by the shield satisfy the test condition, with
$\E[\widehat V(x^{+})\mid x,u]-\widehat V(x)\le-0.018$; $137$ base
proposals are rejected. 
The base policy fails the test on $172$ of
its $181$ kicks and comes to rest as close as $0.1$\,m to an obstacle,
whereas no shielded chain rests closer than $1.33$\,m, as the landing rule requires.

\let\oldthebib\thebibliography
\renewcommand{\thebibliography}[1]{\oldthebib{#1}\fontsize{8}{8.3}\selectfont}
\bibliographystyle{ieeetr}
\bibliography{references}

@book{oksendal2003sde,
  author    = {{\O}ksendal, Bernt},
  title     = {Stochastic Differential Equations: An Introduction with
               Applications},
  edition   = {6th},
  publisher = {Springer},
  year      = {2003}
}

@book{dellacherie1978probabilities,
  author    = {Dellacherie, Claude and Meyer, Paul-Andr{\'e}},
  title     = {Probabilities and Potential},
  publisher = {North-Holland},
  year      = {1978}
}

@article{davis1984pdmp,
  author  = {Davis, M. H. A.},
  title   = {Piecewise-Deterministic {M}arkov Processes: A General Class of
             Non-Diffusion Stochastic Models},
  journal = {Journal of the Royal Statistical Society, Series B},
  volume  = {46},
  number  = {3},
  pages   = {353--388},
  year    = {1984}
}

@inproceedings{hu2000towards,
  author    = {Hu, Jianghai and Lygeros, John and Sastry, Shankar},
  title     = {Towards a Theory of Stochastic Hybrid Systems},
  booktitle = {Hybrid Systems: Computation and Control (HSCC)},
  year      = {2000}
}

@article{bujorianu2006toward,
  author  = {Bujorianu, Manuela L. and Lygeros, John},
  title   = {Toward a General Theory of Stochastic Hybrid Systems},
  journal = {Lecture Notes in Control and Information Sciences},
  year    = {2006}
}

@article{teel2014survey,
  author  = {Teel, Andrew R. and Subbaraman, Anantharaman and
             Sferlazza, Antonino},
  title   = {Stability Analysis for Stochastic Hybrid Systems: A Survey},
  journal = {Automatica},
  volume  = {50},
  number  = {10},
  pages   = {2435--2456},
  year    = {2014}
}

@article{bensoussan1984impulse,
  author    = {Bensoussan, Alain and Lions, Jacques-Louis},
  title     = {Impulse Control and Quasi-Variational Inequalities},
  publisher = {Gauthier-Villars},
  year      = {1984}
}

@article{steinhardt2012finite,
  author  = {Steinhardt, Jacob and Tedrake, Russ},
  title   = {Finite-Time Regional Verification of Stochastic Non-Linear
             Systems},
  journal = {The International Journal of Robotics Research},
  volume  = {31},
  number  = {7},
  pages   = {901--923},
  year    = {2012}
}

@article{santoyo2021barrier,
  author  = {Santoyo, Cesar and Dutreix, Maxence and Coogan, Samuel},
  title   = {A Barrier Function Approach to Finite-Time Stochastic System
             Verification and Control},
  journal = {Automatica},
  volume  = {125},
  year    = {2021}
}

@article{summers2010verification,
  author  = {Summers, Sean and Lygeros, John},
  title   = {Verification of Discrete Time Stochastic Hybrid Systems: A
             Stochastic Reach-Avoid Decision Problem},
  journal = {Automatica},
  volume  = {46},
  number  = {12},
  pages   = {1951--1961},
  year    = {2010}
}

@inproceedings{zikelic2023learning,
  author    = {{\v Z}ikeli{\'c}, Dor{\dj}e and Lechner, Mathias and
               Henzinger, Thomas A. and Chatterjee, Krishnendu},
  title     = {Learning Control Policies for Stochastic Systems with
               Reach-Avoid Guarantees},
  booktitle = {AAAI Conference on Artificial Intelligence},
  year      = {2023}
}

@book{revuz1999continuous,
  author    = {Revuz, Daniel and Yor, Marc},
  title     = {Continuous Martingales and Brownian Motion},
  edition   = {3rd},
  publisher = {Springer},
  year      = {1999}
}

@book{raibert1986legged,
  author    = {Raibert, Marc H.},
  title     = {Legged Robots That Balance},
  publisher = {MIT Press},
  address   = {Cambridge, MA},
  year      = {1986}
}

@article{mathiesen2023safety,
  author  = {Mathiesen, Frederik Baymler and Calvert, Simeon C. and Laurenti, Luca},
  title   = {Safety Certification for Stochastic Systems via Neural Barrier Functions},
  journal = {IEEE Control Systems Letters}, volume = {7}, pages = {973--978}, year = {2023}
}

@inproceedings{neustroev2025neural,
  author    = {Neustroev, Grigory and Giacobbe, Mirco and Lukina, Anna},
  title     = {Neural Continuous-Time Supermartingale Certificates},
  booktitle = {AAAI Conference on Artificial Intelligence}, year = {2025}
}

@book{williams1991probability,
  title     = {Probability with Martingales},
  author    = {Williams, David},
  year      = {1991},
  publisher = {Cambridge University Press},
  address   = {Cambridge, U.K.}
}

@book{kushner1967stochastic,
  author    = {Kushner, Harold J.},
  title     = {Stochastic Stability and Control},
  publisher = {Academic Press},
  address   = {New York},
  year      = {1967}
}

@article{prajna2007framework,
  author  = {Prajna, Stephen and Jadbabaie, Ali and Pappas, George J.},
  title   = {A Framework for Worst-Case and Stochastic Safety Verification
             Using Barrier Certificates},
  journal = {IEEE Transactions on Automatic Control},
  volume  = {52},
  number  = {8},
  pages   = {1415--1428},
  year    = {2007}
}

@inproceedings{lavaei2022shs,
  author    = {Lavaei, Abolfazl and Soudjani, Sadegh and Frazzoli, Emilio},
  title     = {Safety Barrier Certificates for Stochastic Hybrid Systems},
  booktitle = {American Control Conference (ACC)},
  pages     = {880--885},
  year      = {2022}
}

@article{clark2021control,
  author  = {Clark, Andrew},
  title   = {Control Barrier Functions for Stochastic Systems},
  journal = {Automatica},
  volume  = {130},
  pages   = {109688},
  year    = {2021}
}

@article{ames_control_2017,
  author  = {Ames, Aaron D. and Xu, Xiangru and Grizzle, Jessy W. and
             Tabuada, Paulo},
  title   = {Control Barrier Function Based Quadratic Programs for
             Safety Critical Systems},
  journal = {IEEE Transactions on Automatic Control},
  volume  = {62},
  number  = {8},
  pages   = {3861--3876},
  year    = {2017}
}

@article{esfahani2016stochastic,
  author  = {Mohajerin Esfahani, Peyman and Chatterjee, Debasish and
             Lygeros, John},
  title   = {The Stochastic Reach-Avoid Problem and Set Characterization
             for Diffusions},
  journal = {Automatica},
  volume  = {70},
  pages   = {43--56},
  year    = {2016}
}

@article{abate2008probabilistic,
  author  = {Abate, Alessandro and Prandini, Maria and Lygeros, John and
             Sastry, Shankar},
  title   = {Probabilistic Reachability and Safety for Controlled Discrete
             Time Stochastic Hybrid Systems},
  journal = {Automatica},
  volume  = {44},
  number  = {11},
  pages   = {2724--2734},
  year    = {2008}
}

@inproceedings{cosner23robust,
  author    = {Cosner, Ryan K. and Culbertson, Preston and Taylor, Andrew J.
               and Ames, Aaron D.},
  title     = {Robust Safety under Stochastic Uncertainty with Discrete-Time
               Control Barrier Functions},
  booktitle = {Robotics: Science and Systems},
  year      = {2023}
}

@article{cosner_bounding_2024,
  author  = {Cosner, Ryan K. and Culbertson, Preston and Ames, Aaron D.},
  title   = {Bounding Stochastic Safety: Leveraging {F}reedman's Inequality
             with Discrete-Time Control Barrier Functions},
  journal = {IEEE Control Systems Letters},
  volume  = {8},
  year    = {2024}
}

@inproceedings{bansal2017hj,
  author    = {Bansal, Somil and Chen, Mo and Herbert, Sylvia and
               Tomlin, Claire J.},
  title     = {{H}amilton--{J}acobi Reachability: A Brief Overview and
               Recent Advances},
  booktitle = {IEEE Conference on Decision and Control (CDC)},
  pages     = {2242--2253},
  year      = {2017}
}

@article{borquez2026hj,
  author  = {Borquez, Javier and Peng, Shuang and Bansal, Somil},
  title   = {{H}amilton--{J}acobi Reachability for Hybrid Systems: Unified
             Goal-Driven Control with Safety Guarantees},
  journal = {The International Journal of Robotics Research},
  year    = {2026},
  doi     = {10.1177/02783649261477777}
}

@article{grizzle2001asymptotically,
  author  = {Grizzle, Jessy W. and Abba, Gabriel and Plestan, Franck},
  title   = {Asymptotically Stable Walking for Biped Robots: Analysis via
             Systems with Impulse Effects},
  journal = {IEEE Transactions on Automatic Control},
  volume  = {46},
  number  = {1},
  pages   = {51--64},
  year    = {2001}
}

@article{westervelt2003hybrid,
  author  = {Westervelt, Eric R. and Grizzle, Jessy W. and
             Koditschek, Daniel E.},
  title   = {Hybrid Zero Dynamics of Planar Biped Walkers},
  journal = {IEEE Transactions on Automatic Control},
  volume  = {48},
  number  = {1},
  pages   = {42--56},
  year    = {2003}
}

@book{goebel2012hybrid,
  author    = {Goebel, Rafal and Sanfelice, Ricardo G. and Teel, Andrew R.},
  title     = {Hybrid Dynamical Systems: Modeling, Stability, and Robustness},
  publisher = {Princeton University Press},
  year      = {2012}
}

@article{lygeros2003dynamical,
  author  = {Lygeros, John and Johansson, Karl Henrik and Simi{\'c}, Slobodan N.
             and Zhang, Jun and Sastry, S. Shankar},
  title   = {Dynamical Properties of Hybrid Automata},
  journal = {IEEE Transactions on Automatic Control},
  volume  = {48},
  number  = {1},
  pages   = {2--17},
  year    = {2003}
}

\end{document}